\documentclass[12pt]{article}  

\usepackage{graphicx}
\usepackage{amsmath}
\usepackage{amsfonts}
\usepackage{amssymb}
\usepackage{amsthm}
\usepackage{color}
\usepackage{centernot}
\usepackage{mathtools}
\usepackage{stmaryrd}

\newtheorem{theorem}{Theorem}

\begin{document}

\title{Solvable intermittent shell model of turbulence}
\date{Instituto de Matem\'atica Pura e Aplicada -- IMPA, Rio de Janeiro, Brazil \\[10pt]
\today}
\author{Alexei A. Mailybaev\footnote{E-mail: alexei@impa.br}}

\maketitle

\begin{abstract}
We introduce a shell model of turbulence featuring intermittent behaviour with anomalous power-law scaling of structure functions. This model is solved analytically with the explicit derivation of anomalous exponents. The solution associates the intermittency with the hidden symmetry for Kolmogorov multipliers, making our approach relevant for real turbulence.
\end{abstract}

\section{Introduction}

In the ideal fluid dynamics, equations for a velocity field $\mathbf{u}(\mathbf{r},t)$ possess a family of spatiotemporal scaling symmetries of the form  
	\begin{equation}
	\label{eq1}
	t,\mathbf{r},\mathbf{u} \mapsto \alpha^{1-h}t,\alpha\mathbf{r},\alpha^h\mathbf{u}
	\end{equation}
for arbitrary $h \in \mathbb{R}$ and $\alpha > 0$. In the developed turbulence~\cite{frisch1999turbulence}, symmetries (\ref{eq1}) are considered asymptotically at small scales of the so-called inertial interval, where both forcing and viscosity are negligible. Kolmogorov's phenomenological theory~\cite{kolmogorov1941local} proposed that the statistically stationary state of turbulence has symmetry (\ref{eq1}) in the inertial range with the exponent $h = 1/3$ related to a constant energy flux.

It is well-known, however, that the scale invariance (\ref{eq1}) is broken by the phenomenon of intermittency~\cite{frisch1999turbulence,falkovich2009symmetries}: at small scales, real turbulent flows display high activity within small fractions of time or space. Such flows retain some scaling properties described by the power-law relations 
	\begin{equation}
	\label{eq2}
	S_p(\ell) \propto \ell^{\zeta_p},
	\end{equation}
where $S_p(\ell) = \left\langle \left|\delta u_{\parallel}\right|^p\right\rangle$ is a structure function (moment) of order $p$, and $\delta u_{\parallel} = \big[\mathbf{u}(\mathbf{r}+\ell\mathbf{e},t)-\mathbf{u}(\mathbf{r},t)\big]\cdot \mathbf{e}$ is a velocity difference at a distance $\ell$ in a direction $\mathbf{e}$. Here, the average is evaluated with respect to time and does not depend on $\mathbf{r}$ and $\mathbf{e}$ for homogeneous and isotropic turbulence. Unlike Kolmogorov's prediction $\zeta_p = p/3$ following from scaling (\ref{eq1}) with $h = 1/3$, intermittent flows exhibit exponents $\zeta_p$ depending nonlinearly on $p$. This anomalous scaling can be interpreted as a multi-fractal structure of the flow, in which the scale invariance (\ref{eq1}) is restored on a multitude of fractal subsets with fractal dimensions depending on $h$~\cite{frisch1985singularity,frisch1999turbulence}. Theoretical explanation of the anomalous scaling  remains a major open problem in turbulence. 

Shell models of turbulence~\cite{gledzer1973system,ohkitani1989temporal,l1998improved,biferale2003shell} are simplified caricatures of fluid dynamics, which possess scaling symmetries analogous to (\ref{eq1}) and, possibly, some inviscid invariants resembling physical conservation laws. These models are constructed for complex scalar variables $u_n(t)$ with integer indices $n$, which mimic velocity fluctuations (or Fourier modes) corresponding to scales $\ell_n \sim \lambda^{-n}$. One calls $n$ a shell number and $u_n$ a shell velocity. Conventionally, we define the wavenumbers  $k_n = 1/\ell_n = \lambda^n$ and take $\lambda = 2$. Though shell models are discrete, the set of scales $\ell_n$ satisfies the scaling relation 
	\begin{equation}
	\label{eq3b}
	\ell_{n+m} = \alpha\ell_n
	\end{equation}
for any
	\begin{equation}
	\label{eq3a}
	\alpha = \lambda^{-m}, \quad m \in \mathbb{Z}.
	\end{equation}
Thus, scaling symmetries (\ref{eq1}) can be reformulated for shell models as
	\begin{equation}
	\label{eq3}
	t,u_n \mapsto \alpha^{1-h}t,\alpha^hu_{n+m}
	\end{equation}
for arbitrary $h \in \mathbb{R}$ and $\alpha$ of the form (\ref{eq3a}). The structure functions are defined as $S_p(\ell_n) = \left\langle |u_n|^p\right\rangle$. In numerical studies, shell models successfully reproduce the phenomenon of intermittency: power laws (\ref{eq2}) are observed at small scales $\ell = \ell_n$ of the inertial interval with the exponents $\zeta_p$ depending nonlinearly on $p$ and being close to their values in real turbulence~\cite{biferale2003shell}. However, similarly to the full model, theoretical explanation of intermittency in shell models remains open. 

In the present work, we propose a solvable intermittent shell model: it allows the rigorous analysis of intermittency with explicit computation of anomalous exponents. For this purpose, we pursue the celebrated conception of making the model ``as simple as possible, but no simpler''. Unlike previous shell models, we relax the requirement of energy conservation: the energy imposes some restrictions on scaling exponents (e.g., the value $\zeta_3 = 1$ is attributed to the energy cascade), but it does not seem to be intrinsic for the anomaly. Indeed, theoretical attempts to explain intermittency using universal multipliers~\cite{kolmogorov1962refinement,benzi1993intermittency,eyink2003gibbsian}, multi-fractality~\cite{frisch1985singularity} or zero-modes~\cite{falkovich2001particles,biferale2003shell} do not rely on energy conservation. At the same time, symmetries (\ref{eq3}) must be preserved because they are immanent in the phenomenon of intermittency. 

Our model is designed such that it can be reduced to a system of decoupled equations for Kolmogorov multipliers (velocity ratios) using a hierarchy of intrinsic times at different scales of motion. This construction follows the ideas of~\cite{benzi1993intermittency,eyink2003gibbsian} and illustrates the concept of hidden scaling symmetry developed in~\cite{mailybaev2021hidden,mailybaev2020hidden}. We introduce the model in Section~\ref{sec_M}. Main results are formulated and proved in Section~\ref{sec_R}. Section~\ref{sec_N} presents a numerical demonstration. In the last Section~\ref{sec_D}, we discuss the hidden symmetry and the relevance of our approach for real turbulence. 

\section{Model}\label{sec_M}

In our shell model, the complex shell variables $u_n(t)$ satisfy the equations
	\begin{equation}
	\label{eq4}
	\frac{du_n}{dt} = k_n|u_n|^2 F_n
	+u_n\sum_{m < n}\mathrm{Re}\left(k_mu_m^*F_m\right),
	\end{equation}
where 
	\begin{equation}
	\label{eq5}
	F_n = \frac{|u_{n-1}|^2}{|u_n|^2} f\left(\frac{u_n}{|u_{n-1}|}\right)
	\end{equation}
with the function $f(z)$ specified later in (\ref{eq9}). For any complex function $f(z)$, all quantities $F_n$ in (\ref{eq5}) are homogeneous functions of degree 0 with respect to shell velocities and, hence, the right-hand side of system (\ref{eq4}) is homogeneous of degree 2. As a consequence, one can see that system (\ref{eq4}) possesses all the symmetries (\ref{eq3}). 
Notice that though the right-hand side of (\ref{eq4}) has the required scaling properties, it is not a quadratic form in shell velocities. Furthermore, functions (\ref{eq5}) are singular for vanishing shell variables. This is no problem: we will see that none of variables vanish in our model.
 
As in the theory of turbulence~\cite{frisch1999turbulence}, we must modify our model at large scales by introducing the so-called forcing range. We restrict time-dependent variables $u_n(t)$ to the shells $n = 1,2,\ldots$ and set 
	\begin{equation}
	\label{eq5bn}
	u_0(t) \equiv 1
	\end{equation}
at the integral scale $\ell_0$. Then, equations in (\ref{eq4}) are considered for $n > 0$ with the sum restricted to the range $0 < m < n$. 

We remark that the model is closed for a truncated collection of shell variables $u_1,\ldots,u_N$ for any $N > 0$. This property is crucially used in our analysis, which is based on the theory of finite-dimensional dynamical systems. Unlike most conventional shell models, we do not need viscous terms because our model does not conserve energy. Also, the sum in equation (\ref{eq4}) contains nonlocal  coupling to the modes at large-scales $\ell_m \gg \ell_n$: this is the property of realistic flows that is not present in usual shell models. 

We choose the complex function $f$ in (\ref{eq5}) written in polar coordinates as
	\begin{equation}
	\label{eq9}
	f(z) = \left(R(\theta)+\frac{dR(\theta)}{d\theta}-\rho
	+i\rho\right)e^{i\theta},\quad
	z = \rho e^{i\theta},
	\end{equation}
where $R(\theta)$ is any $C^2$ positive real function on the circle $\theta \in S^1\, (\mathrm{mod}\ 2\pi)$. Additionally, we require that $R(\theta)$ satisfies the condition
	\begin{equation}
	\label{eq9C}
	R(\theta)+\frac{dR(\theta)}{d\theta} > 0\quad \textrm{for}\quad \theta \in S^1,
	\end{equation}
which facilitates the proofs below.

\section{Intermittency}\label{sec_R}

Solutions of the proposed shell model are described by the following two theorems. The first theorem states the existence, uniqueness and boundness of solutions. 	
 
\begin{theorem}\label{th1}
For any initial condition with nonzero shell velocities $u_n(0) \ne 0$, $n > 0$, there exists a unique solution $u_n(t)$ of the shell model for $t \ge 0$. There exist constants $0 < c_n < C_n$ depending on initial conditions, such that
	\begin{equation}
	\label{eq_bound}
	c_n \le |u_n(t)| \le C_n \quad \textrm{for} \quad t \ge 0.
	\end{equation}
\end{theorem}

\begin{proof}
Equations (\ref{eq4}) form a decoupled system for any set of shell variables $u_1,\ldots,u_N$. Hence, it is sufficient to prove the theorem by considering such a finite-dimensional subsystem for arbitrary $N$.
The right-hand sides in this subsystem are continuously differentiable and bounded if conditions (\ref{eq_bound}) are satisfied. This guarantees the existence and uniqueness of solution for $t \ge 0$. Therefore, we will prove the theorem by verifying conditions (\ref{eq_bound}).

Let us introduce a multiplier $w_n$ and a rescaled time $\tau_n$ at each shell $n$ as
	\begin{equation}
	\label{eq6}
	w_n = \frac{u_n}{|u_{n-1}|},\quad d\tau_n = k_n|u_{n-1}|dt,
	\end{equation}
with $\tau_n = 0$ at $t = 0$.
There is one-to-one correspondence between the shell velocities $u_n$ and multipliers $w_n$, provided that all of them are nonzero. This correspondence is defined by expressions (\ref{eq5bn}) and (\ref{eq6}) with the inverse relations given by
	\begin{equation}
	\label{eq16}
	u_1 = w_1, \quad u_n = w_n\prod_{m = 1}^{n-1}|w_m|.
	\end{equation}
	
Using (\ref{eq6}), we obtain
	\begin{equation}
	\label{eq7}
	\begin{array}{l}
	\displaystyle
	\frac{dw_n}{d\tau_n} 
	= \frac{1}{k_{n}|u_{n-1}|}\frac{d}{dt}\frac{u_n}{|u_{n-1}|} 
	\displaystyle
	= \frac{1}{k_{n}|u_{n-1}|^2}\frac{du_n}{dt}
	-\frac{u_n}{k_{n}|u_{n-1}|^3}\frac{d|u_{n-1}|}{dt}
	\\[15pt]
	\displaystyle
	\qquad = \frac{1}{k_{n}|u_{n-1}|^2}\frac{du_n}{dt}
	-\frac{u_n}{k_{n}|u_{n-1}|^4}\,\mathrm{Re}\left(u_{n-1}^*\frac{du_{n-1}}{dt}\right).
	\end{array}
	\end{equation}
Substituting equations (\ref{eq4}) and (\ref{eq5}) into (\ref{eq7}), a long but straightforward derivation yields the simple relation
	\begin{equation}
	\label{eq8}
	\frac{dw_n}{d\tau_n} = f(w_n).
	\end{equation}
This equation can be written using polar coordinates $w_n = \rho_n e^{i\theta_n}$ and expression (\ref{eq9}) as
	\begin{eqnarray}
	\label{eq8PCa}
	\displaystyle
	\frac{d\rho_n}{d\tau_n} 
	& = & 
	\displaystyle
	\mathrm{Re}\left[e^{-i\theta_n}f(w_n)\right] = R(\theta_n)+\frac{dR(\theta_n)}{d\theta_n}-\rho_n,
	\\[7pt]
	\label{eq8PCb}
	\displaystyle
	\frac{d\theta_n}{d\tau_n} 
	& = & 
	\displaystyle
	\frac{\mathrm{Im}\left[e^{-i\theta_n}f(w_n)\right]}{\rho_n} = 1.
	\end{eqnarray}
Solving these equations, we find a general solution in the form
	\begin{equation}
	\label{eq8r}
	w_n(\tau_n) = \left[r_ne^{-\tau_n}+R(\tau_n+\varphi_n)\right]e^{i(\tau_n+\varphi_n)},
	\end{equation}
where $\varphi_n$ and $r_n$ are defined by the initial condition as
	\begin{equation}
	\label{eq8rIC}
    \varphi_n = \arg w_n(0),\quad 
    r_n = |w_n(0)|-R(\varphi_n).
    \end{equation}
Expression (\ref{eq8r}) yields the upper bound
	\begin{equation}
	\label{eq11U}
	|w_n(\tau_n)| \le C'_n 
	\end{equation}
for $\tau_n \ge 0$ and  $C'_n = |r_n|+\max_\theta R(\theta)$. Due to condition (\ref{eq9C}), the derivative (\ref{eq8PCa}) is positive for $\rho_n < R_*$ and arbitrary $\theta_n$, where we introduced a positive constant $R_* = \min_\theta\left[R(\theta)+dR/d\theta\right]$. Hence, by taking $c'_n = \min\{\rho_n(0),R_*\} > 0$, we have
	\begin{equation}
	\label{eq11Ub}
	c'_n \le \rho_n(\tau_n) = |w_n(\tau_n)|
	\end{equation}
for $\tau_n \ge 0$. Expressions (\ref{eq16}) with (\ref{eq11U}) and (\ref{eq11Ub}) imply the analogous bounds (\ref{eq_bound}) for the variables $u_n$. \end{proof}

Now we formulate and prove the main result:
	
\begin{theorem}\label{th2}
Let us assume that the mean value
    \begin{equation}
	\label{eqDF2}
	R_0 = \frac{1}{2\pi}\int_0^{2\pi}
	R(\theta)d\theta
	\end{equation}
is a transcendental number. Then, for any initial condition with nonzero shell velocities $u_n(0) \ne 0$, $n > 0$, the structure function 
	\begin{equation}
	\label{eq13b}
	S_p(\ell_n) = \lim_{T \to +\infty} \frac{1}{T} \int_0^T  |u_n(t)|^p dt
	\end{equation}
of any order $p \in \mathbb{R}$ has the power-law form
	\begin{equation}
	\label{eq13d}
	S_p(\ell_n) = \ell_n^{\zeta_p}
	\end{equation}
with the exponent 
	\begin{equation}
	\label{eq13exp}
	\zeta_p = -\log_2 \left(\frac{1}{2\pi}
	\int_0^{2\pi} 
	R^p(\theta)d\theta\right).
	\end{equation}
\end{theorem}

\begin{proof}
Let us consider the shell model equations for $u_1,\ldots,u_N$ with arbitrary fixed $N$. As we showed in the proof of Theorem~\ref{th1}, this is a decoupled subsystem. We will now derive the invariant probability measure for this system, prove its unique ergodicity, and use this property for computing the structure functions. 

Let us consider the dynamics for multipliers $w_1,\ldots,w_N$, which are uniquely related to $u_1,\ldots,u_N$. According to (\ref{eq8r}), the system has a globally attracting invariant manifold given by the limit cycles 
	\begin{equation}
	\label{eq10}
	w_n = R(\theta_n)e^{i\theta_n},
	\quad
	n = 1,\ldots,N,
	\end{equation}
where $\theta_n = \tau_n+\varphi_n$.
Using (\ref{eq10}) in (\ref{eq16}), we represent this invariant manifold in terms of original shell variables as
	\begin{equation}
	\label{eqM1}
	u_n = e^{i\theta_n}\prod_{m = 1}^n R(\theta_m),
	\quad
	n = 1,\ldots,N,
	\end{equation}
which is parametrized by the set of phase variables $(\theta_1,\ldots,\theta_N)$ on the $N$-dimensional torus $\mathbb{T}^N = (S^1)^N$. Using (\ref{eqM1}) in the second expression of (\ref{eq6}), we write equation (\ref{eq8PCb}) in terms of the original time as
	\begin{equation}
	\label{eqM3}
	\frac{d\theta_n}{dt} = k_n\prod_{m = 1}^{n-1}R(\theta_m).
	\end{equation}

Let us introduce the new coordinates $(\alpha_1,\ldots,\alpha_N) \in  \mathbb{T}^N$ on the torus as
	\begin{equation}
	\label{eqM4}
	\alpha_n = \theta_n+h_n(\theta_1,\ldots,\theta_{n-1}),\quad
	n = 1,\ldots,N,
	\end{equation}
with continuously differentiable functions $h_n:\mathbb{T}^{n-1}\mapsto S^1$. 
Below we will prove that these functions can be chosen such that equations (\ref{eqM3}) in new variables take the form
	\begin{equation}
	\label{eqM3b}
	\frac{d\alpha_n}{dt} = k_nR_0^{n-1}.
	\end{equation}
Since $R_0$ is transcendental and $k_n$ is integer, the numbers 
	\begin{equation}
	\label{eqM3bR}
	k_1,k_2R_0,\ldots,k_NR_0^{N-1} 
	\end{equation}
in the right-hand sides of (\ref{eqM3b}) are rationally independent. By the Kronecker--Weyl equidistribution theorem~\cite{katok1997introduction}, system of equations (\ref{eqM3b}) for $n = 1,\ldots,N$ is uniquely ergodic with the uniform (Lebesgue) invariant measure on the torus $\mathbb{T}^N$. Since the Jacobian for the change of coordinates (\ref{eqM4}) is unity, the same is true for system (\ref{eqM3}). Using (\ref{eqM1}), we write 
	\begin{equation}
	\label{eq13bB}
	|u_n|^p = \prod_{m = 1}^n R^p(\theta_m).
	\end{equation}
By the ergodic theorem~\cite{katok1997introduction}, for any initial conditions with nonzero components, the limit in (\ref{eq13b}) is equal to the average with respect to angles:
	\begin{equation}
	\label{eq13bA}
	S_p(\ell_n) 
	= \frac{1}{(2\pi)^n}\int\cdots\int \prod_{m = 1}^n R^p(\theta_m)d\theta_m 
	= \left(\frac{1}{2\pi}\int R^p(\theta)d\theta\right)^n.
	\end{equation}
This expression yields (\ref{eq13d}) and (\ref{eq13exp}) for $\ell_n = 2^{-n}$.

It remains to find the proper functions $h_n$. Let us write 
	\begin{equation}
	\label{eqM6}
	R(\theta) = R_0+R_1(\theta),
	\end{equation}
where $R_1(\theta)$ is a periodic function with zero mean value. 
Using (\ref{eqM3}) and (\ref{eqM4}), we obtain
	\begin{equation}
	\label{eqM5}
	\frac{d\alpha_n}{dt} = \frac{d\theta_n}{dt}
	+\sum_{m = 1}^{n-1}\frac{\partial h_n}{\partial \theta_m}
	\frac{d\theta_m}{dt}
	= k_n\prod_{m = 1}^{n-1}R(\theta_m)
	+\sum_{m = 1}^{n-1}\frac{\partial h_n}{\partial \theta_m}
	\left(k_m\prod_{j = 1}^{m-1}R(\theta_j)\right).
	\end{equation}
One can check that the first product in the right-hand side can be written using (\ref{eqM6}) as
	\begin{equation}
	\label{eqM5P}
	\prod_{m = 1}^{n-1}R(\theta_m)
	= R_0^{n-1}+\sum_{m = 1}^{n-1}
	\left(\prod_{j = 1}^{m-1}R(\theta_j)\right)R_1(\theta_m)R_0^{n-1-m}.
	\end{equation}
Using (\ref{eqM5P}) in (\ref{eqM5}) yields
	\begin{equation}
	\label{eqM5P2}
	\frac{d\alpha_n}{dt} = k_nR_0^{n-1}+\sum_{m = 1}^{n-1}\left(
	k_n R_1(\theta_m) R_0^{n-1-m}
	+k_m\frac{\partial h_n}{\partial \theta_m}\right)\prod_{j = 1}^{m-1}R(\theta_j).
	\end{equation}
Let us define 
	\begin{equation}
	\label{eqM6b2}
	h_n(\theta_1,\ldots,\theta_{n-1}) = -\sum_{m = 1}^{n-1}
	k_{n-m}R_0^{n-1-m}Q(\theta_m),\quad
	Q(\theta) = \int_0^\theta R_1(\theta')d\theta',
	\end{equation}
where $Q(\theta)$ is a periodic function because $R_1(\theta)$ has zero mean.
One can verify that equation (\ref{eqM5P2}) with $h_n$ from (\ref{eqM6b2}) yields (\ref{eqM3b}). \end{proof}

\section{Numerical tests}\label{sec_N}

One can verify that $\zeta_p$ given by formula (\ref{eq13exp}) is a concave function of $p$ with the properties
	\begin{equation}
	\label{eqN1P}
	\zeta_0 = 0,\quad \lim_{p \to +\infty}\frac{\zeta_p}{p} = -\log_2\left[\max_{\theta\in S^1} R(\theta)\right],\quad
	\lim_{p \to -\infty}\frac{\zeta_p}{p} = -\log_2\left[\min_{\theta\in S^1} R(\theta)\right].
	\end{equation}
By selecting different functions $R(\theta)$, one can obtain different forms of $\zeta_p$ of this kind.
As a specific example, let us consider
	\begin{equation}
	\label{eqN1}
	R(\theta) = R_0\left(1+\frac{\cos \theta}{2}\right),\quad
	R_0 = \frac{2\pi}{9}.
	\end{equation}
Figure~\ref{fig1}(a) shows the graph of $\zeta_p$ obtained numerically by formula (\ref{eq13exp}). One can see that the exponents depend nonlinearly on $p$, therefore, demonstrating the anomalous scaling behaviour. For integer values of $p$, one can find these exponents analytically using known expressions for integrals of trigonometric functions as 
	\begin{equation}
	\label{eqN2a}
	\zeta_p = -p\log_2R_0-\log_2I_p,\quad
	I_p = \frac{1}{2\pi}\int_0^{2\pi} \left(1+\frac{\cos \theta}{2}\right)^pd\theta
	= \sum_{k = 0}^{[p/2]}  \frac{p!16^{-k}}{k!k!(p-2k)!}.
	\end{equation}

\begin{figure}[t]
\centering
\includegraphics[width=0.8\textwidth]{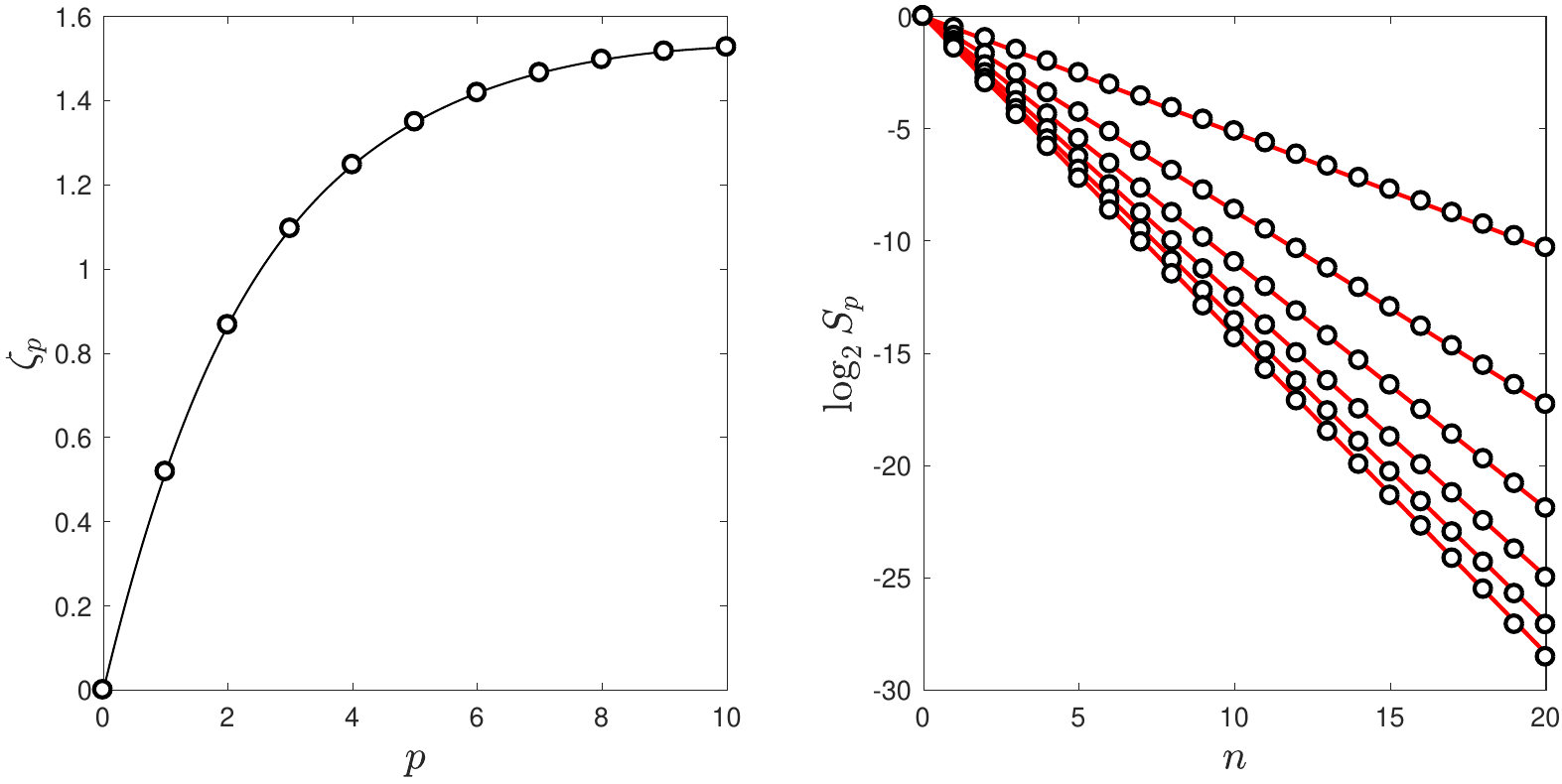}
\caption{(a) Anomalous exponents $\zeta_p$ of structure functions computed by formula (\ref{eq13exp}). (b) Structure functions (\ref{eq13b}) obtained by the numerical simulation for the orders $p = 1,\ldots,6$ (empty dots) are compared with the analytical prediction (\ref{eq13d}) and (\ref{eqN2a}) (red lines).}
\label{fig1}
\end{figure}

We verify this result by integrating equations (\ref{eq4})--(\ref{eq5bn}) and (\ref{eqN1}) of the shell model numerically with the initial condition $u_n(0) = k_n^{-1/3}e^{in}$ for the shells $n = 1,\ldots,20$. The numerical simulation was carried out with high accuracy in the large time interval $\Delta t = 10100$. Initial times $t < 100$ were ignored in the computation of structure functions (\ref{eq13b}) in order to reduce transient effects. The results are shown by empty circles in Fig.~\ref{fig1}(b) for $p = 1,\ldots,6$ in logarithmic coordinates: $\log_2 S_p$ vs. $n = -\log_2\ell_n$. This figure verifies that power laws (\ref{eq13d}) with exponents (\ref{eqN2a}) shown by red lines are undistinguishable from the numerical results up to tiny statistical fluctuations.

\begin{figure}
\centering
\includegraphics[width=0.73\textwidth]{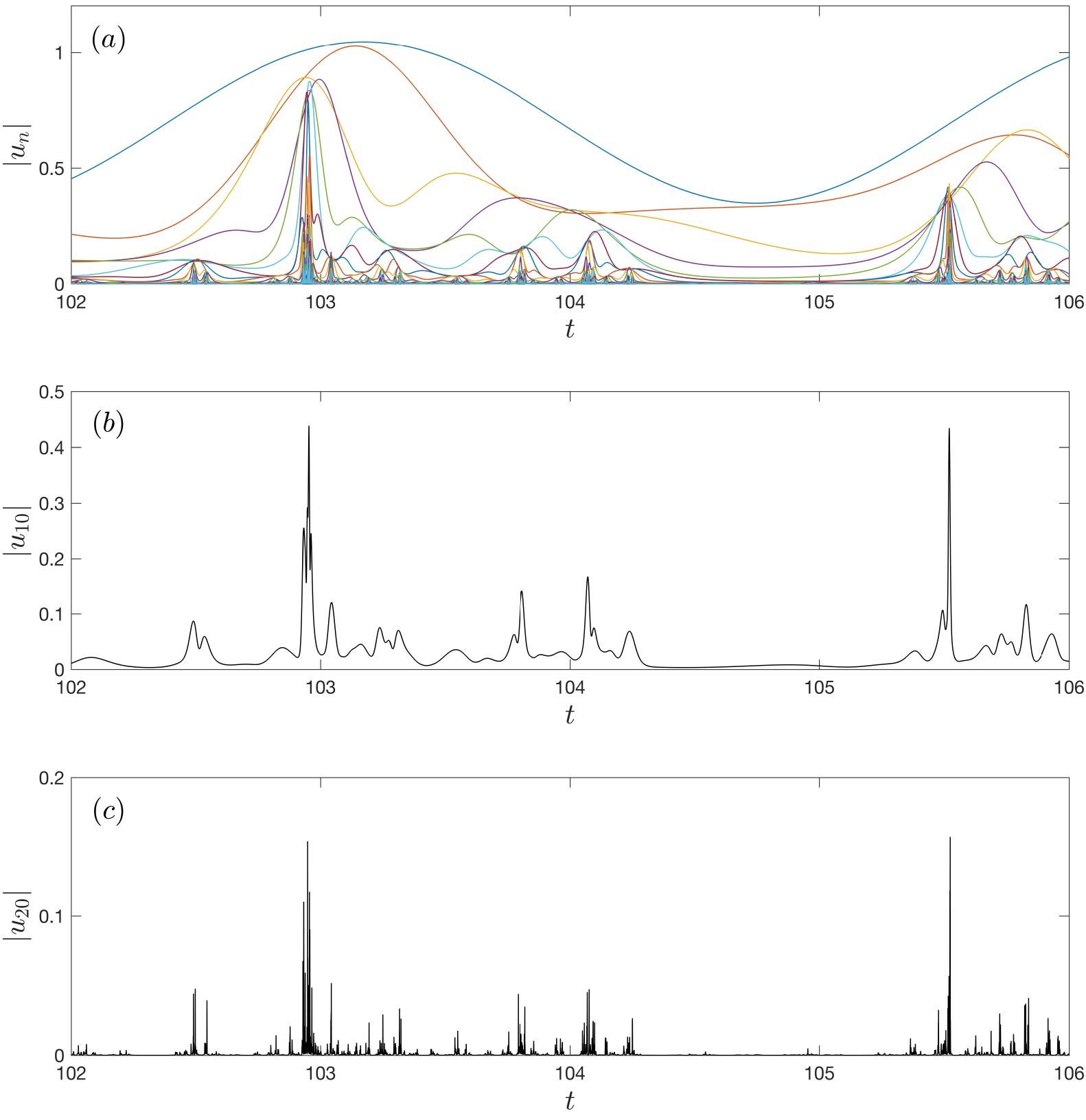}
\caption{Evolution of absolute values of shell variables $|u_n(t)|$ obtained by the numerical simulation: (a) for the shells $n = 1,\ldots,20$, (b) $n = 10$ and (c) $n = 20$. One can see that the intermittency gets stronger at smaller scales (larger shell numbers).}
\label{fig2}
\end{figure}

Figure~\ref{fig2} shows typical temporal behaviour for absolute values of shell variables $|u_n|$ as functions of time: the panel (a) presents all shell variables, (b) $n = 10$ and (c) $n = 20$. These figures demonstrate the intermittent dynamics, which is intrinsic to the anomalous scaling: for larger shell numbers $n$ (smaller scales $\ell_n = 2^{-n}$), the dynamics features much shorter intervals of large-amplitude ``turbulent'' oscillations separated by much longer intervals of low-activity ``laminar'' behaviour.

\section{Discussion: hidden symmetry and intermittency}\label{sec_D}

We presented a solvable intermittent shell model of turbulence, in which anomalous exponents of structure functions are computed analytically. In this final section, we develop a deeper insight into the origin of anomalous scaling in this  model. Our ability to solve the equations of motion is based on the representation (\ref{eq8}) for the Kolmogorov multipliers $w_n$ written as functions of proper temporal variables $\tau_n$. This representation is self-consistent: using (\ref{eq6}), relations among different times are expressed in terms of multipliers as
	\begin{equation}
	\label{eqHS1}
	d\tau_{n+m} 
	= d\tau_n \times \left\{\begin{array}{ll}
	k_m|w_n \cdots w_{n+m-1}|, & m > 0; \\[5pt]
	k_m/|w_{n+m}\cdots w_{n-1}|,& m < 0.
	\end{array}
	\right.
	\end{equation}
Each multiplier $w_n$ oscillates periodically with respect to the corresponding time $\tau_n$. Then, the mutual dynamics of multipliers as functions of original time $t$ becomes quasi-periodic with rationally independent frequencies; see Eq.~(\ref{eqM3b}). For the truncated system with variables $u_1,\ldots,u_N$,  the attractor dimension is equal to $N$. Therefore, the attractor in our shell model (with $N = \infty$) is genuinely infinite-dimensional.

The quasi-periodicity allows computing 
the structure functions using the ergodic theorem with the uniform invariant measure in the space of phases $\theta_n = \arg w_n$. The property of the measure to be uniform (more importantly, universal with respect to shell numbers $n$) is crucial for the power-law scaling of structure functions. Clearly, this universality follows from the universality of representation (\ref{eq8}), which is identical for all shell numbers $n$. 
The universality of multipliers just described can be interpreted as a new symmetry. The system of equations (\ref{eq8}) and relations (\ref{eqHS1}) among different times are symmetric with respect to the change
	\begin{equation}
	\label{eqHS2}
	w_n \mapsto w_{n+m}, \quad \tau_n \mapsto \tau_{n+m},
	\end{equation}
applied to all shells $n$ with a fixed but arbitrarily chosen shift $m \in \mathbb{Z}$. Transformations (\ref{eqHS2}) define a new (hidden) scaling symmetry of equations of motion, which substitutes the original scaling relations (\ref{eq3}). In the intermittent dynamics, the stationary probability distribution remains symmetric with respect to the hidden scaling symmetry, despite all the original scaling symmetries are broken. This is exactly the property, which yields the anomalous power-law scaling in our system.

The hidden symmetry (\ref{eqHS2}) was analysed in the context of the Sabra shell model of turbulence  in~\cite{mailybaev2021hidden}, and in \cite{mailybaev2020hidden} we showed that this symmetry follows from the non-commutativity of temporal scalings with an evolution operator. Furthermore, extension of this theory to systems with Galilean invariance (Galilean transformations also do not commute with an evolution operator) yields analogous hidden scaling symmetries in realistic models of fluid dynamics~\cite{mailybaev2020hidden}, such as the Euler equations for incompressible ideal fluid~\cite{mailybaev_thalabard2021hidden}. It was also shown in \cite{mailybaev2020hidden} that hidden-symmetric probability distributions yield power-law scaling of structure functions with exponents, which are typically anomalous. The shell model of the present work provides a detailed analytic demonstration of the hidden symmetry and its relation to intermittency, therefore, justifying its possible relevance for real turbulence.

Finally, let us comment on the role of conservation laws such as the conservation of energy, which we ignored in our shell model. The energy cascade is an important constituent of real turbulence, which relates the energy input at large forcing scales with the energy dissipation at small viscous scales. For shell models, the energy is usually defined as $E = \sum |u_n|^2$. Since $|u_n| = |w_1\cdots w_n|$, the conservation of energy at small scales (large $n$) establishes relations among multipliers at distant scales. This explains why equations of motion for the multipliers become coupled~\cite{eyink2003gibbsian} but still maintaining the hidden symmetry (\ref{eqHS2})~\cite{mailybaev2021hidden}. This coupling greatly complicates the analysis but does not change the situation conceptually: the stationary probability distribution remains hidden-symmetric and this symmetry leads to the intermittency, as demonstrated numerically in  \cite{mailybaev2020hidden}. In this case, the multipliers are not statistically independent in a spirit of Kolmogorov's theory of 1962~\cite{kolmogorov1962refinement}, but rather have correlations decaying at distant scales~\cite{benzi1993intermittency,eyink2003gibbsian}. Last but not least, one may expect that the coupling of multipliers induces chaotic and spontaneously stochastic behaviours~\cite{mailybaev2016spontaneously,thalabard2020butterfly,drivas2020statistical}. Further developments of our solvable model may help in understanding the mutual importance of these  effects in turbulence.

\textbf{Acknowledgments: } The author is grateful to Artem Raibekas for his help in the study of ergodicity, and to Theodore D. Drivas, Simon Thalabard and the anonymous reviewer for their comments on the manuscript. The work is supported by CNPq (grants 303047/2018-6, 406431/2018-3).

\bibliographystyle{plain}
\bibliography{refs}

\end{document}